\documentclass[journal]{IEEEtran}
\ifCLASSINFOpdf
\else
\fi
\usepackage{lipsum}
\usepackage{cuted}
\usepackage{amssymb}
\usepackage{mathtools}
\usepackage{enumerate}
\usepackage{hhline}
\usepackage{bm}
\usepackage{enumitem}
\usepackage{graphicx}
\usepackage{subfigure}
\usepackage{xcolor}
\usepackage{color, soul} 
\usepackage{tcolorbox}
\tcbuselibrary{skins, breakable, theorems,fitting}
\usepackage{soul}
\usepackage{cite}
\usepackage{todonotes}
\usepackage{algpseudocode}
\usepackage[utf8]{inputenc}

\usepackage{algorithm}
\algdef{SE}[SUBALG]{Indent}{EndIndent}{}{\algorithmicend\ }%
\algtext*{Indent}
\algtext*{EndIndent}

\newtheorem{theorem}{Theorem}

\newtheorem{remark}{Remark}

\newtheorem{proof}{Proof }
\newtheorem{corollary}{Corollary}

\begin{document}

\title{Nested Bayesian Optimization for Computer Experiments}

\author{Yan Wang,~\IEEEmembership{}
        Meng Wang,~\IEEEmembership{}
        Areej AlBahar,~\IEEEmembership{}
        Xiaowei Yue, ~\IEEEmembership{IEEE Senior Member}

\thanks{Manuscript received XX, 2021; revised XX, 2021.
\textit{(Corresponding author: Xiaowei Yue)}. Dr. Yue's research was partially supported by the National Science Foundation (2035038) and the Grainger Frontiers of Engineering Grant Award from the National Academy of Engineering (NAE); Dr. Wang's research was supported by the Natural Science Foundation
of Beijing Municipality (1214019).}
\thanks{Y. Wang and M. Wang are with the School of Statistics and Data Science, Faculty of Science, Beijing University of Technology, Beijing 100124, China. (e-mail: yanwang@bjut.edu.cn)}
\thanks{A. AlBahar and X. Yue are with the Grado Department of Industrial and Systems Engineering, Virginia Tech, Blacksburg, VA, 24061 USA (e-mail: areejaa3@vt.edu; xwy@vt.edu)}

}

\markboth{Manuscript to IEEE/ASME Trans}{}

\maketitle
\setcounter{page}{1}
\begin{abstract}
Computer experiments can {emulate the physical systems, help computational investigations, and yield analytic solutions. They have been widely employed with many engineering applications (e.g., aerospace, automotive, energy systems).} Conventional Bayesian optimization did not incorporate the nested structures in computer experiments. This paper proposes a novel nested Bayesian optimization {method} for complex computer experiments with multi-step or hierarchical characteristics. 
{We prove the theoretical properties of nested outputs given that the distribution of nested outputs is Gaussian or non-Gaussian.}
The closed forms of nested expected improvement are derived. We also propose the computational algorithms for nested Bayesian optimization. Three numerical studies show that the proposed nested Bayesian optimization { method outperforms the five benchmark Bayesian optimization methods that ignore the intermediate outputs} of the inner computer code. The case study shows that the nested Bayesian optimization can efficiently minimize the residual stress during composite structures assembly and {avoid convergence to local optima}.

\end{abstract}


%
\IEEEpeerreviewmaketitle

\section{Introduction}
\label{intro}
%
%
%
%

\IEEEPARstart{C}{OMPUTER}{ experiments have become increasingly used in engineering simulations due to the development of information technology and computing power}. Especially for the scenarios where physical experiments are difficult, expensive, or impossible to implement, computer experiments can serve as proxy surrogates for and adjuncts to physical experiments \cite{santner2018design}. In advanced manufacturing and mechatronics, typical computer experiments may rely on Finite Element Analysis (FEA), Computational Fluid Dynamics (CFD), multiphysics simulation, variation propagation analysis, etc. Widely used engineering simulation software includes ANSYS, Matlab/Simulink, COMSOL Multiphysics, Solidworks, 3DCS. Sophisticated computer codes can model the multi-step or multi-physics processes accurately, thereby improving the efficiency of engineering design, system optimization, and quality control.

\subsection{Nested Computer Experiments}
Firstly, we will illustrate what is \emph{nested computer experiment}, and why the nested effect is very critical for engineering simulations, in particular for advanced manufacturing. {If one model or system contains the outputs of the other model or system, we call them \emph{nested}.} Nested property usually comes from the hierarchical structures of systems and multiphysics phenomena. In practice, one system often contains a few subsystems; the output of one subsystem could be the input for the sequential subsystem. Nested structures are ubiquitous in engineering simulation. {Suppose one computer experiment includes multi-layer sequential operations/codes, and outputs from one computer code may serve as the inputs for the other level of computer code. In that case, we call it a nested computer experiment.} The \emph{nested computer experiment codes} are also called  \emph{System of Solvers} in engineering. 

Most computer simulations and digital twins for multistage manufacturing processes (MMP) are nested, {because of the natural multi-step structure and inherent hierarchy in advanced manufacturing systems.} In MMP, multiple operations/stations are involved to produce one product \cite{shi2006stream}, as shown in Fig. \ref{figsov}. The product quality variations can propagate from one station to its downstream station. Stream of Variation methodologies have been developed to model and reduce the variation and improve the quality control \cite{shi2006stream,zhang2016stream}. When simulating the MMP in Fig. \ref{figsov}, the inputs for stage $k$ include two types: input quality features $\bm{q}_{k-1}$ from the upstream stage $k-1$, and the new process-induced deviations and noise at the current stage. Similarly, the outputted quality features  $\bm{q}_{k}$ of Stage $k$ will also serve as inputs for downstream stage $k+1$. Wen et al. developed a computer simulation for composite aircraft assembly process \cite{wen2018feasibility,wen2019virtual}, where the simulation needs multiple steps even for a single-stage assembly, as shown in Fig. \ref{Figure flow}. Therefore, the omnipresent nested structure needs to be incorporated when modeling computer experiments.

\begin{figure}[t!]
\centering
\includegraphics[width=\columnwidth]{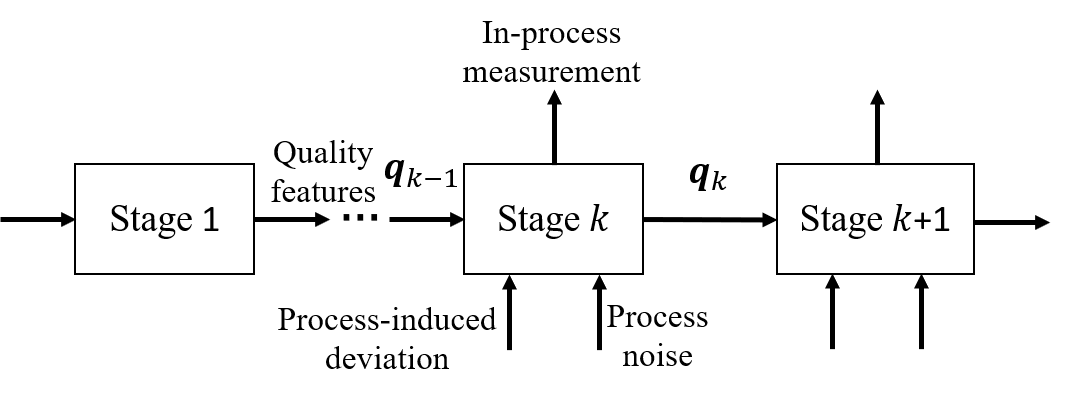}
\caption{Variation Propagation in Multistage Manufacturing Systems.}
\label{figsov}
\end{figure}



\subsection{Literature Review}
In this section, {we conduct the literature review from three fields: mechatronics, advanced statistics, and manufacturing systems.}

In the mechatronics field, Rodriguez et al. developed one hybrid control scheme with two nested loops for twisted string actuators \cite{rodriguez2020hybrid}. Nested design techniques have been used for co-design of controlled systems \cite{kamadan2017co}. Zeng et al. proposed a nested optimization strategy to guarantee cost control for a motor driving system \cite{zeng2019integrated}. The performance-based nested Kriging model was constructed to interpolate the Antenna characteristics data \cite{koziel2019performance}. Nested long-short term memory (LSTM) networks were incorporated into deep learning architecture for multivariate air quality prediction \cite{jin2021multivariate}. A nested tensor product model transformation was used to analyze the Takagi-Sugeno fuzzy system for system control design \cite{yu2018nested}. These approaches make full use of the nested structure for various objectives (control, design, prediction, etc.) and achieve excellent performance.

In the advanced statistics field, researchers investigated nested effects in computer experiments. Nested space-filling designs were constructed for computer experiments with two levels of simulation accuracy \cite{qian2009nested}. Next, nested Latin hypercube designs with sliced structures were proposed for experimental data collection \cite{chen2015nested}. Hung et al. developed the optimal Latin hypercube designs and kriging methods incorporating nested factors and branching factors \cite{hung2009design}. Marque-Pucheu et al. proposed an efficient dimension reduction method for Gaussian process emulation of two nested codes \cite{marque2020efficient}. Keogh and White investigated nested case-control and case-cohort study on exposure-disease association \cite{keogh2013using}. These methods significantly improve the efficiency and effectiveness of data collection, model emulation, and association analysis in advanced statistics.

In the advanced manufacturing field, nested systems have also been investigated. Gibson et al. used multivariate nested distributions to model semiconductor process variability \cite{gibson1999using}. Similarly, Tian et al. analyzed the nested variation pattern in the batch processes of semiconductor manufacturing, and proposed a two-level nested control chart for process monitoring \cite{tian2016two}. Jin and Shi developed a reconfigured piecewise linear regression tree to model the nested structure for process control in multistage manufacturing \cite{jin2012reconfigured}. Savin and Vorochaeva developed a quadratic programming based controller with nested structure, and it achieved excellent performance in planar pipeline robots \cite{savin2017nested}. Wang et al. proposed multiresolution and multisensor fusion network for fault diagnosis, with integration of multiple network structures \cite{wang2019multilevel}. {These methods enhanced variability modeling, process control, and quality assurance by accommodating the nested structure.}

\subsection{Novelty and Contributions}

Although numerous techniques have been investigated in studying and using nested effect, as mentioned in the literature review above, global optimization for nested computer experiments {still lacks a systematic science base.} This paper focuses on the global optimization of nested computer experiments. We mainly use \emph{two-layer nested computer models} as one example for nested computer experiments. The first-layer code is denoted as the inner computer model, and the second one as the outer computer model. The nested structure indicates that the outputs of the inner computer model are part of inputs of the outer computer model. The inner computer model and outer computer model are very complex and they are assumed to be black-box.

{Bayesian optimization is an efficient approach to obtain the global optimal solution for complex computer experiments given specific objectives.} {This approach has proven to be successful in many real-world engineering optimization problems, such as  the robust parameter design} \cite{tan2020bayesian}, {the multi-objective optimization problems} \cite{shu2020new,biswas2022multi,mathern2021multi}, { the constrained optimization problems }\cite{tran2019constrained}. 
The main steps of a standard Bayesian optimization method include: (i) Build a statistical surrogate model based on previous computer outputs; (ii) Choose an acquisition function and sequentially query the objective function at points which maximize the acquisition. For step (i), the most popular stochastic surrogate model  is the Gaussian Process (GP) model \cite{santner2018design}. 
For step (ii), commonly used acquisition functions include the Expected Improvement (EI) \cite{jones1998efficient,ranjan2013comment}, the Lower/Upper
Confidence Bound (LCB)  \cite{srinivas2010gaussian}, and the Expected Quantile Improvement (EQI) acquisition functions \cite{picheny2013quantile}.  
Despite the wide applications of Bayesian optimization methods,
these existing methods ignored the outputs of the inner computer model and treated all the inputs characterizing the system of interest as a single input vector. When trying to find the global optimal solution of nested computer experiments, these existing Bayesian optimization methods are less efficient, since the nested structure information is ignored in the optimization. Astudillo and Frazier \cite{astudillo2019bayesian} considered Bayesian optimization of composite functions and took the outputs of the inner  part of a composite function into account. This method performs excellent when the outer part of a composite function is a known, cheap-to-evaluated, and real-valued function. It does not work well for the complex black-box functions with nested structure, which is more common in engineering computer experiments.

{In this work, we proposed a novel and systematic Bayesian optimization {method} for nested computer experiments. We assume that both the inner and outer computer models are deterministic, but expensive-to-evaluate}. Our contributions can be summarized as follows: 
 \begin{itemize}
   \item  The nested Bayesian optimization {method} is proposed to incorporate the nested structures in complex computer experiments. This method can learn the global optimum more efficiently and avoid convergence to the local optimum. 
   \item  {We investigated the theoretical properties of the nested Gaussian process for two cases: 1. it can be approximated by a Gaussian process and 2. it cannot be approximated by a Gaussian process.} Furthermore, we derive the closed forms of nested expected improvement and propose a computational algorithm for nested Bayesian optimization.  
   \item  Based on the composite structures assembly case study, we show {that nested Bayesian optimization can minimize the residual stress after assembly. We also show the proposed nested Bayesian optimization performs better than five benchmark methods via numerical studies.} 
\end{itemize}
The outline of this paper is as follows: Section \ref{sec.prob} introduces the optimization problem of two-nested computer experiments.  Section \ref{sec.meth} proposes the nested Bayesian optimization method.  Section \ref{sec.num} and Section \ref{case} compare the proposed method with the standard Bayesian optimization method {by using three numerical studies and a real case study}.  Concluding remarks are given in Section \ref{sec.dis}. Appendices  contain {detailed proofs of the theorems and selection of correlation functions}

\section{Problem setting}
\label{sec.prob}

In this section, we use mathematical models to describe the problem setting. Denote $f:\mathcal{X} \rightarrow\mathbf{R}$ to be a nested computer model, which is defined as 
\begin{equation}
f(\Tilde{\bm x})=g(\bm h^T(x), x'); \Tilde{\bm x}=(x,x')^T\in\mathcal{X}\subset \mathbf{R}^d,
\end{equation}
where 
$\bm h(x)=\left(h_1(x),\ldots,h_p(x)\right)^T, p\geq 1$ is a
vector of inner computer model outputs.  $g(\cdot)$  is the outer computer model whose inputs include outputs of the inner computer model $\bm h(x)$ and the additional control variable $x'$. {There is a  serial relationship between the inner computer model and outer computer model. Intermediate outputs $\bm h(x)$ and $x'$ are parallel inputs. }
Fig. \ref{Figure nce} shows the framework of nested computer experiments:
\begin{figure}[h!]
\centering
\includegraphics[width=\columnwidth]{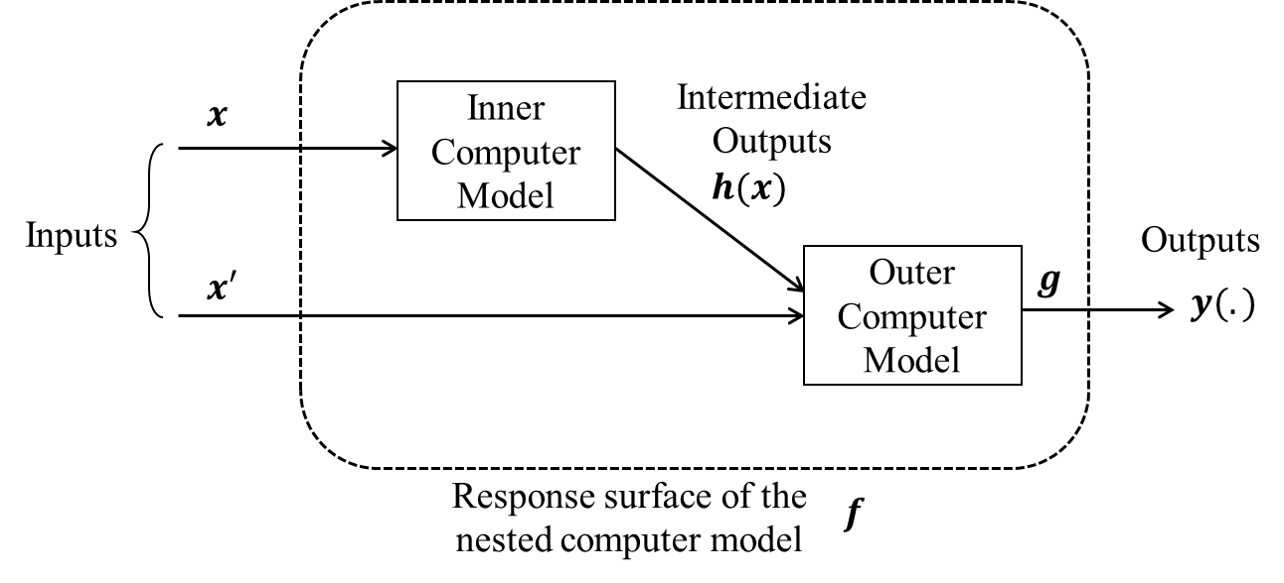}
\caption{Nested computer experiments.}
\label{Figure nce}
\end{figure}

Suppose {these two computer models are black-box, deterministic, expensive-to-evaluate}, and the gradient information is not available.  With the help of a limited number of  outputs from both computer models, we consider the problem of finding a minimizer of the entire response surface of the nested computer model $f$: 
\begin{equation}
\label{goal}
\Tilde{\bm x}^{*}=\operatorname*{argmin}_{\Tilde{\bm x}\in\mathcal{X}}f(\Tilde{\bm x}).
\end{equation}

 {Specifically, suppose the nested computer experiments are conducted at the points $\tilde{X}_n=(\Tilde{\bm x}_1,\ldots, \Tilde{\bm x}_n)^T$, which  contains the collections of $\{\bm x_1,\ldots,\bm x_n \}$ and   $\{\Tilde{\bm x}_1,\ldots,\Tilde{\bm x}_n \}$. The first-layer computer model generates intermediate outputs $H_n=(\bm h(x_1),\ldots,\bm h(x_n))^T$, and the second-layer computer model generates the outputs $Y_n=\left(g(\bm h^T(x_1), x'_1),\ldots,g(\bm h^T(x_n),x'_n)\right)^T$. These computer experiments yield data $D_n=\{\tilde{X}_n,H_n,Y_n\}$. The goal of this work is to query $\Tilde{\bm x}^{*}$ by making full use of  the dataset $D_n$.
}

As discussed above, the standard Bayesian optimization method  can be used to solve the optimization problem (\ref{goal}). This approach can query the optimal point of $f$ sequentially  by optimizing an acquisition function. {In this work, we focus on the EI criterion} \cite{jones1998efficient,santner2018design}. { Detailed comparisons are conducted between EI, LCB, and EQI-based approaches in Section} \ref{sec.num} { and Section} \ref{case}.

 The main idea of EI  is to sample the point offering the greatest
expected improvement over the current best sampled point. Let $f_n^*=\min_{i=1}^n\{y_i\}$ be the current best objective value,  given data
$\{\tilde{X}_n,Y_n\}$, the EI  function becomes:
\begin{equation}
\label{ei}
{\rm EI}_n(\Tilde{\bm x})=E_{f|\tilde{X}_n,Y_n}(f_n^*-f(\Tilde{\bm x}))_{+},
\end{equation}
where 
$(f_n^*-f(\Tilde{\bm x}))_{+}=\max\{f_n^*-f(\Tilde{\bm x}),0\}$  is the improvement utility function. 

It can be known that the evaluation of EI depends on  the posterior distribution ${f|\tilde{X}_n,Y_n}$.   Since the posterior distribution ${f|\tilde{X}_n,Y_n}$ in standard Bayesian optimization method ignores the outputs of the inner computer model, {it leads to low optimization efficiency or even getting stuck in a local optimum when the number of samples is limited.} 
To overcome this limitation, we will develop a new Bayesian optimization method to incorporate the nested structure and identify the optimal solution for complex computer experiments.

\section{Nested Bayesian optimization }
\label{sec.meth}
{Nested computer experiments are ubiquitous when running engineering simulations, digital twin or finite element analysis. Conventional Bayesian optimization approaches consider the entire system as a whole and try to identify the global optimum for black-box functions. They are less efficient in complex systems optimization when nested structures exist. The nested structures usually can be determined according to the system configurations or engineering knowledge. By incorporating the nested structures of complex systems, we can make full use of more information in Bayesian optimization, intuitively avoid getting stuck in some local optima, and have the potential to improve optimization efficiency.}
In this section, we propose a novel method, named as \emph{Nested Bayesian Optimization (NBO)}{, to  query  the global optimal solution of  nested computer experiments}. To approximate the outputs of nested computer experiments,  we first introduce nested Gaussian Process (NGP) models in Section \ref{sec:NGP}. Next, we derive the closed forms of the expected improvement  acquisition function for nested computer experiments in Section \ref{sec:nei}, under the cases that the NGP models are Gaussian and non-Gaussian.  Section \ref{sec:alg} provides a detailed algorithm of the NBO method.

\subsection{Nested Gaussian Process models}
\label{sec:NGP}
In this work, Gaussian Process (GP) models \cite{santner2018design} are used to mimic the inner and the outer computer models.
Suppose $h$ and $g$ are realizations of two Gaussian Processes. 
Given data $D_n$, the posterior distribution of the inner computer model at an unobserved input $x$ is
\begin{equation}
\label{posnestin}
\bm h(x)| D_n\sim N(\hat {\bm h}_{n}(x),\bm s_{h}^2(x)),
\end{equation}
where $\hat {\bm h}_{n}(x)$ is a $p\times 1$ mean vector, and $\bm s_{h}^2(x)$ is a $p\times p$ covariance matrix. The posterior distribution of the outer computer model at an unobserved input $x^{out}=(\bm h^T, x')$  is
\begin{equation}
\label{posnest2}
g(x^{out})|D_n \sim N(\hat g_n(x^{out}),s_g^2(x^{out})).
\end{equation}
 Formulations of the posterior mean and posterior variance function are given by (\ref{blue})  and (\ref{var}), respectively. 
  More Details about the Gaussian Process models can be found in Appendix \ref{gp model}.


The \emph{nested Gaussian Process {(NGP)} model} is expressed as 
\begin{equation}
\label{nested gp}
 f(\Tilde{\bm x})|D_n=\hat g_n({\Psi}^T(x),x')+s_g({\Psi}^T(x),x')\xi_g.
 \end{equation}
where $\Psi(x)={\bm h}(x)|D_n$, $\xi_g$ is a standard normal random variable. From the posterior distribution of the inner computer model (\ref{posnestin}), $\Psi(x)$ can be represented as 
$\Psi(x)=\hat {\bm h}_n(x)+\bm s_h(x) \bm \xi_h$,
 where $ \bm \xi_h$  is a $p\times 1$  random vector that follows the normal distribution and it is independent from $ \xi_g$. 
By numerical calculations, we have that,
 the posterior variance of $f(\tilde{\bm x})|{D_n}$ is zero for any $i=1,\ldots,n$, and the posterior mean  is interpolating the observed data values $(\tilde{X}_n,Y_n)$.

From (\ref{nested gp}), we can see that $\Psi(x) $ obeys a normal distribution when $\bm s_h(x)\neq 0$. As a function of $\Psi(x) $, the posterior distribution of $f(\tilde{\bm x})|{D_n}$ may not be normal.
 Therefore, we will investigate two cases, Gaussian and non-Gaussian in the following part.
 
 Theorem \ref{thm-1} focuses on the Gaussian case, while Theorem \ref{lemma2} analyzes the non-Gaussian case. 



\begin{theorem}
\label{thm-1}
Denote $\mu_{Z}(\Tilde{\bm x})=\hat g_n(\hat {\bm h}^T_n(x),x')$ and $s^2_{Z}(\Tilde{\bm x})=s^2_g(\hat {\bm h}^T_n(x),x')$. The NGP model (\ref{nested gp}) { is the following GP model}
\begin{equation}
\label{gp case}
\begin{aligned}
 GP(\mu_{Z}(\Tilde{\bm x}),s^2_{Z}(\Tilde{\bm x})),
 \end{aligned}
\end{equation}
if and only if for all $\Tilde{\bm x}\in \mathcal{X}$,  {there is  
$\bm s_h(x)= \bm 0_{1\times p}$.}
\end{theorem}

For ease of understanding, here we give the brief proof of Theorem \ref{thm-1}. First,  $\bm s_h(x)= \bm 0_{1\times p}$  indicates that {the surrogate of inner computer model is deterministic}. By plugging $\Psi( x)=\hat{\bm h}_n(x)$ into (\ref{nested gp}),{ we can derive that  $f(\tilde{\bm x})|{D_n}$ obeys a normal distribution for fixed $\tilde{\bm x}$. 
In addition, the NGP model is gaussian, implying that at least one of the following two conditions holds:}
\begin{itemize}
    \item  {The outer computer model is independent on the inner computer outputs, i.e., the NGP model} (\ref{nested gp}) {can be expressed as $\hat g_n(x')+s_g(x')\xi_g$. Due to the nested structure, both $\hat g_n$ and $s_g$ depend on $\Psi$. This condition  is not true.
 
  }
    \item  {$\Psi( x)=\hat{\bm h}_n(x)$. It indicates that $\bm s_h(x)$ equals to zero and the {surrogate of inner computer model is deterministic}.}
\end{itemize}

 Theorem \ref{thm-1} {states that for a nested computer model, the NGP is a GP model if and only if  {the surrogate of inner computer  model is deterministic}. This condition is hard to achieve or even unattainable in some cases.  Indeed, from Corollary} \ref{cor-gp} {, when $\bm s_h$ is close to $\bm 0$, i.e., the inner GP model can achieve satisfactory prediction accuracy,   the  GP model} (\ref{gp case}) { can be used to mimic the nested computer experiments.}


\begin{theorem}
\label{lemma2}
{Denote} $\bm c^T_h(\Tilde{\bm x})=\frac{\partial \hat g_n}{\partial \bm h}(\hat {\bm h}^T_n(x),x')\bm s_h(x)$,  $c_g(\Tilde{\bm x})=s_g(\hat {\bm h}^T_n(x),x')$, and $\bm c^T_{h,g}(\Tilde{\bm x})=\frac{\partial  s_g}{\partial \bm h}(\hat {\bm h}^T_n(x),x')\bm s_h(x)$. {Assume that the second order derivatives of $\hat g_n$ and $s_g$ with respect to $\bm h$  are uniformly bounded.} 
{The NGP model} (\ref{nested gp}) is a non-Gaussian Process model if and only if there is  $ \Tilde{\bm x}\in \mathcal{X}$, {such that   $\bm s_h(x)\neq \bm 0_{1\times p}$.} Specifically, in this case, the NGP model (\ref{nested gp}) {can be approximated by}
\begin{equation}
\label{ngp-general}
\begin{aligned}
Z(\Tilde{\bm x})=Z_1(\Tilde{\bm x})Z_2(\Tilde{\bm x})+z_0(\Tilde{\bm x}).
 \end{aligned}
\end{equation}
{Here,   $Z_1(\Tilde{\bm x})$ and $Z_2(\Tilde{\bm x})$ are  independent Gaussian Processes with mean functions $\mu_1(\Tilde{\bm x})={c_g(\Tilde{\bm x})}/\sqrt{\bm c^T_{h,g}(\Tilde{\bm x})\bm c_{h,g}(\Tilde{\bm x})}$, $\mu_2(\Tilde{\bm x})=\sqrt{\bm c^T_{h}(\Tilde{\bm x})\bm c_{h}(\Tilde{\bm x})}$ respectively and variance functions $\sigma^2_1(\Tilde{\bm x})=1$, $\sigma^2_2(\Tilde{\bm x})={\bm c^T_{h,g}(\Tilde{\bm x})\bm c_{h,g}(\Tilde{\bm x})}$ respectively; $z_0(\Tilde{\bm x})=\mu_Z(\Tilde{\bm x})- \mu_1(\Tilde{\bm x})\mu_2(\Tilde{\bm x})$. In addition,  the mean and variance functions  of $Z(\Tilde{\bm x})$ are }

\begin{equation}
\label{post-ngp}
\begin{aligned}
{\rm{E}}[Z(\Tilde{\bm x})]=&\mu_Z(\Tilde{\bm x})=\hat g_n(\hat {\bm h}^T_n(x),x'),\\
{\rm Var}[Z(\Tilde{\bm x})]=&\bm c^T_h(\Tilde{\bm x})\bm c_h(\Tilde{\bm x})+c_g^2(\Tilde{\bm x})+\bm c^T_{h,g}(\Tilde{\bm x})\bm c_{h,g}(\Tilde{\bm x}).
 \end{aligned}
\end{equation}
\end{theorem}

\begin{remark}
{For a fixed $\Tilde{\bm x}\in  \mathcal{X}$,  $Z(\Tilde{\bm x})$ is a non-Gaussian random variable. 
The exact probability density function of  $Z(\Tilde{\bm x})$ is given by} (\ref{exactpdf}). {If  $z_0(\Tilde{\bm x})=0$,  $Z(\Tilde{\bm x})$ follows a normal product (NP) distribution} \cite{2016Exact},  {which is in general non-Gaussian. Especially,  if $Z_1(\Tilde{\bm x})\sim N(0,1)$ and $Z_2(\Tilde{\bm x})\sim N(0,1)$, then density function of $Z_1(\Tilde{\bm x})Z_2(\Tilde{\bm x})$ is}
\begin{equation}\nonumber
\displaystyle p_{Z}(z)= \frac{K_{0}({|z|})}{\pi },\infty <z<+\infty.
\end{equation}
{Here $K_{0}$ denotes the modified Bessel function of the second kind with order ${0}$. This density function exhibits a sharp peak at the origin and heavy tails. }
 \end{remark} 

 {Detailed proof of Theorem} \ref{lemma2}  {can be found in Appendix} \ref{App:proof}.
Theorem \ref{lemma2} {states that the NGP model can be approximated by a non-Gaussian process model $Z(\Tilde{\bm x})$. The global trend of $Z(\Tilde{\bm x})$ is the same as the posterior mean of } (\ref{nested gp}). The variance of  $Z(\Tilde{\bm x})$ involves three kinds of uncertainty: $\bm c^T_h(\Tilde{\bm x})=\frac{\partial \hat g_n}{\partial \bm h}(\hat {\bm h}^T_n(x),x')\bm s_h(x)$ is the uncertainty due to the inner GP model; $c_g(\Tilde{\bm x})=s_g(\hat {\bm h}^T_n(x),x')$ is the uncertainty   due to the outer GP model; $\bm c^T_{h,g}(\Tilde{\bm x})=\frac{\partial  s_g}{\partial \bm h}(\hat {\bm h}^T_n(x),x')\bm s_h(x)$ is the uncertainty arising from the combined effect of the inner and outer  models. In addition, from Theorem \ref{lemma2}, we have that, there is a great difference between the NGP and composite GP \cite{ba2012composite}. The composite GP model is an addition of two Gaussian Processes, where the first
one captures the smooth global trend and the second one models local details. Thus the composite GP is still a Gaussian Process. However, the NGP may no longer be a Gaussian Process.

\begin{corollary}
\label{cor-gp}
 If $\bm s_h(x)$ {converges to} $\bm 0_{1\times p}$ for all  $\Tilde{\bm x}\in \mathcal{X}$, $Z_2(\Tilde{\bm x})$ tends to be a deterministic function. In this case the model (\ref{ngp-general}) {converges to the GP model} (\ref{gp case}).
\end{corollary}

Corollary \ref{cor-gp} shows that the NGP model (\ref{nested gp}) {can be approximated by the  GP model} (\ref{gp case}), if  $\bm s_h(x)$ is small for all  $\Tilde{\bm x}\in \mathcal{X}$. It relaxes the condition  for an NGP model able to be approximated by a GP model in Theorem \ref{thm-1}.

{From Theorem  }\ref{thm-1} {and Theorem } \ref{lemma2}{, we can see that, the posterior mean and variance function of the NGP model depend only on the posterior mean and variance of the inner and the outer GP models. Given the fact that the computational complexity for the outer GP model is $O(n^3)$,  and for the inner GP model is $O\left((pn)^3)\right)$ }\cite{santner2018design}, the computational complexity for the NGP model is $O\left((pn)^3)\right)$.

\subsection{ Closed forms of the Nested Expected Improvement (NEI)}
\label{sec:nei}
To distinguish from the standard Bayesian optimization method, the EI function where NGP is used to approximate the nested computer experiments is called  \emph{{Nested Expected Improvement (NEI)}} function:
\begin{equation}
\label{nei}
{\rm NEI}_n(\Tilde{\bm x})=E_{f|D_n}(f_n^*-f(\Tilde{\bm x}))_{+},
\end{equation}
A new queried point ${\Tilde{\bm x}}_{n+1}$  is  selected by maximizing the  ${\rm NEI}_n(\Tilde{\bm x})$ function
\begin{equation}
\label{maxnei}
{\Tilde{\bm x}}_{n+1}= \operatorname*{argmax}_{\Tilde{\bm x}\in\mathcal{X}}{\rm NEI}_n(\Tilde{\bm x}).
\end{equation}
We can see that values of ${\rm NEI}_n$ depend on the posterior distribution ${f(\Tilde{\bm x})|D_n}$. Given two cases depending on whether NGP model can be approximated by a Gaussian process, the NEI acquisition function also has different expressions. Specifically,
\begin{itemize}
\item {If the NGP model can be approximated by the GP model} (\ref{gp case}){, denote $v(\Tilde{\bm x})=\frac{f_n^*-\mu_Z(\Tilde{\bm x})}{s_Z(\Tilde{\bm x})}$,  the NEI acquisition function has the closed-form expression:}
\begin{eqnarray}
\label{neigp}
(f_n^*-\mu_Z(\Tilde{\bm x}))
{\rm\Phi}_N\left(v(\Tilde{\bm x})\right)+ s_Z(\Tilde{\bm x})\phi_N\left(v(\Tilde{\bm x})\right).
\end{eqnarray}


\item {If the NGP model cannot be approximated by a GP model, the NEI acquisition function can be evaluated by:}
\begin{eqnarray}
\label{ngp-ei}
\small
\begin{aligned}
\int_{-\infty}^{\infty}&{\left(f_n^*-z_0(\Tilde{\bm x})-t\mu_2(\Tilde{\bm x})\right)}\phi_N \left(u_1(t,\Tilde{\bm x})\right)\Phi_N\left(u_2(f_n^*,t,\Tilde{\bm x})\right)\\
 & + |t|\sigma_2(\Tilde{\bm x}) \phi_N \left(u_1(t,\Tilde{\bm x})\right)\phi_N\left( u_2(f_n^*,t,\Tilde{\bm x})\right)dt.
\end{aligned}
\end{eqnarray}
\end{itemize}

{Detailed derivation of} (\ref{neigp}) {and} (\ref{ngp-ei}) {can be found in}  \cite{jones1998efficient} {and  Appendix}  \ref{App:proof}{, respectively.}

\begin{remark}
{The NEI acquisition function} (\ref{neigp}) {implicitly encodes
a tradeoff between exploration of the feasible region and exploitation near the current best solution. The first term in} (\ref{neigp}) {encourages exploitation, by assigning larger values for points  with smaller predicted values; the second
term in} (\ref{neigp}) {encourages exploration, by assigning greater values for points with
larger estimated posterior variance. } 
\end{remark}
\begin{remark}
 {Markov Chain Monte Carlo (MCMC) method can be used to estimate ${\rm NEI}_n$} (\ref{ngp-ei}). {Because $\phi_N\left(u_1(t,\Tilde{\bm x})\right)=0$ as $u_1(t,\Tilde{\bm x})$ tends to infinity, the interval of integration $t\in (-\infty,\infty)$ can be shrunk to $t\in [L_t(\Tilde{\bm x}),U_t(\Tilde{\bm x})]$, where  $L_t(\Tilde{\bm x})$ and $U_t(\Tilde{\bm x})$ are  pre-specified, such as $L_t=-10\sigma_1(\Tilde{\bm x})+\mu_1(\Tilde{\bm x})$ and $U_t=10\sigma_1(\Tilde{\bm x})+\mu_1(\Tilde{\bm x})$ respectively.}   
\end{remark}



\begin{remark} 
 Sampled Expected Improvement (SEI) as suggested in  \cite{chen2020finding} is  a commonly used method to {estimate} EI values when $f(\Tilde{\bm x})|D_n$ is non-Gaussian.   SEI estimates EI values based on  a large number of posterior samples of $f(\Tilde{\bm x})|D_n$ and only the prediction posterior samples that are smaller than the current best value are taken in the calculation.  Since generating posterior samples of $f(\Tilde{\bm x})|D_n$ by using the posterior density function  (\ref{exactpdf}) is  rather time-consuming, this method loses attraction.

\end{remark}


\subsection{Algorithm}
\label{sec:alg}

In this subsection, we develop the computational algorithm for nested Bayesian optimization. Algorithm \ref{alg:1}  provides detailed steps of the  NBO method.

\begin{algorithm}
\caption{Nested Bayesian optimization}
\label{alg:1}
\begin{algorithmic}[1]
    \State Obtain an initial design $\tilde{X}_{n_0}$ with $n_0$ points, and run the nested computer models at these points, yielding corresponding simulator outputs $H_{n_0},Y_{n_0}$.
    \For {iteration \texttt{$n=n_0,\cdots,N-1$}}
        \State Evaluate the current best optimal point  ${\Tilde{\bm x}}_n^*=$  \Indent $\operatorname*{argmin} Y_n$ and the corresponding function value 
        \EndIndent
        \Indent $f_n^*=\min Y_n$.
        \EndIndent
        \State Build GP models (\ref{posnestin}) and (\ref{posnest2}) to mimic the inner and 
        \Indent the outer computer models respectively.
        \EndIndent
        \State Test whether the NGP model is a GP model by using 
             \Indent a cross-validation method.
        \EndIndent
            \If {NGP model is Gaussian}
                \State Identify the maximizer $\Tilde{\bm x}_{n+1}$ of ${\rm NEI}_n$ (\ref{neigp}).
            \Else
                \State Identify the maximizer $\Tilde{\bm x}_{n+1}$ of ${\rm NEI}_n$ (\ref{ngp-ei}).
            \EndIf  
        \State Run the nested  computer models  at $\Tilde{\bm x}_{n+1}$, augment 
        \Indent  $\tilde{X}_n$, $H_n$ and $Y_n$ with $\Tilde{\bm x}_{n+1}$, $h(x_{n+1})$ and $f( \Tilde{\bm x}_{n+1})$.
        \EndIndent
    \EndFor
    \State \textbf{Return} the current best optimal point ${\Tilde{\bm x}}_N^*=  \operatorname*{argmin} Y_N$ and the corresponding function value $f_N^*=\min Y_N$. 
    \end{algorithmic}
\end{algorithm}

We can explain this algorithm as follows. Firstly, initial data is collected based on
a  maximin Latin hypercube design. Here, the number of initial points $n_0$ is set at $10d$, as recommended in \cite{Loeppky2009Special}.
 Next, Gaussian Process models  are built to mimic the inner model and the outer model by using  (\ref{posnestin}) and (\ref{posnest2}). 
Then, $K$-fold cross-validation method is used to exam whether the NGP is a GP or not. More specifically, build GP model (\ref{gp case}) to approximate the nested computer outputs and then examine the prediction accuracy of this GP model by $K$-fold cross-validation method. Here, choice of $K$ follows the criterion below \cite{jung2018multiple}
$$K\approx \log(n) \textit{ and } n/K>3d.$$
Finally, query the sequential points by maximizing (\ref{neigp})  (when NGP is Gaussian) or by maximizing  (\ref{ngp-ei}) (when NGP is non-Gaussian), until the sample size budget $N$ is reached.


\section{Numerical studies}
\label{sec.num}

{
In this section, we compare the  proposed NEI method with five benchmark methods. The five benchmark methods include (1) EI-GP: the Expected Improvement (EI) method under the one-GP model; (2) LCB-GP: the Lower Confidence Bound (LCB) method under the one-GP model; (3) LCB-NGP: the Lower Confidence Bound (LCB) method under the NGP model; } 
{ (4) EQI-GP: the Expected Quantile Improvement (EQI) method under the one-GP model; and (5) EQI-NGP: the Expected Quantile Improvement (EQI) method under the NGP model. }  {The tuning parameter for the LCB function is selected following }\cite{srinivas2010gaussian,brochu2010tutorial}.  

{
The simulation set-up is as follows.
We generate the inputs $\tilde X_{n_0}$,where $n_0=10d$, according to a maximin Latin hypercube design via the R package \emph{maximinLHS}. Then, we collect the inner computer model outputs  $H_{n_0}$,  and the outer computer model outputs $Y_{n_0}$ on $H_{n_0}$ and $\tilde X_{n_0}$ . 

To obtain the NGP predictor, two GP models are built to mimic the inner and outer computer models, respectively.  Here, the GP models are fitted using the R package \emph{DiceKriging}
} \cite{roustant2015package}.

{The log-optimality gap is used to compare the performance of different methods, which is defined as} 
$$ log_{10}(f_n^*-f^*).$$
{All results about the log-optimality gap are averaged over 50 replications.}

\subsection{ 1-d GP model}
\label{1dgp}
Suppose the inner computer model and the outer computer model are both commonly used one-dimension test functions in the literature on GP models \cite{santner2018design}:
   \begin{equation}
   \begin{aligned}
 h( x)&=\exp(-1.4x)\cos(7\pi x/2)-1.4x, x\in[0,1],\\\nonumber
 g( h)&=h\sin(\pi h/2).
 \end{aligned}
  \end{equation}
The global minimum of $f(x)=g(h({ x}))$ is at $ x^*=0.124$ and the corresponding function value  is $0$.

{By choosing the Gaussian correlation functions }(\ref{exp}) { as the correlation functions}, two GP models are built to mimic the inner and outer computer models, respectively. {To illustrate the reasons we the choose Gaussian correlation functions, a detailed comparison of the model accuracy between the one-GP model and the NGP model under different correlation functions is given in Appendix} \ref{rbfs}. Fig. \ref{Fig. pred} shows  predictors and 95\% confidence intervals given by these two GP models. 

\begin{figure}[!ht]
\centering
\includegraphics[width=\columnwidth]{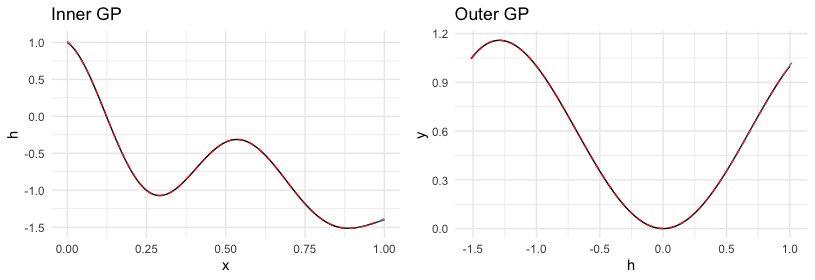}
\caption{{The true inner and outer computer models (black real lines) v.s. GP predictions (red dotted lines) and 95\% confidence intervals (blue intervals) of the inner computer model (left) and outer computer model (right).}}
\label{Fig. pred}
\end{figure}

From Fig. \ref{Fig. pred}, we can find that the inner and outer computer models can be approximated by GP models perfectly. Moreover, 95\% confidence intervals of the inner GP predictor show that $s_h( x)$ is almost zero for all $x\in [0,1]$. Therefore, $f(\Tilde{\bm x})$ can be approximated by a GP model. To further verify this conclusion, a Gaussianity test is then conducted. 

 By the $3$-fold cross-validation (CV) method, we have that, the NGP model is a GP model. Fig. \ref{Figure ngp} compares the performance of the one-GP build by using $(\tilde X_{n_0},Y_{n_0})$ and the NGP model approximated by a composite GP model. It can be seen that, both mean functions of the one-GP model and the NGP model  match the true function accurately, but the 95\% confidence intervals indicate that, the NGP predictor has smaller variance than the one-GP predictor.

 {The reason for this result is that, $f$ is a realization from a  non-stationary GP. Compared to the stationary one-GP model, the NGP model can approximate $f$ more accurately and can also improve
the prediction intervals, especially  when the
experimental design is sparse} \cite{ba2012composite}.

\begin{figure}[!ht]
\centering
\includegraphics[width=\columnwidth]{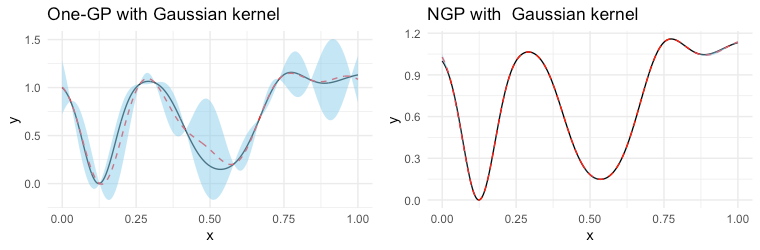}
\caption{{Left:  predictions (red dotted line) and 95\% confidence intervals  of the one-GP model build by using $(\tilde X_{n_0},Y_{n_0})$, with $n_0=10$;  Right: predictions (red dotted line) and 95\% confidence intervals of the NGP model.}}
\label{Figure ngp}
\end{figure}

 Fig.  \ref{Figure 1} shows the log-optimality gap against the
number of samples for the six methods.

\begin{figure}[!ht]
\centering
\includegraphics[width=\columnwidth]{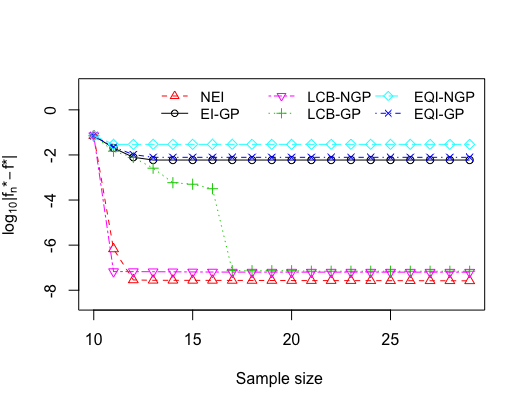}
\caption{Average optimality gap over $50$ replications by different methods.}
\label{Figure 1}
\end{figure}

{From Fig.} \ref{Figure 1}{, we can see that, the optimality gaps for NEI, LCB-NGP and LCB-GP  enjoy steady improvements as $n$ increases, whereas the optimality gap for the other methods stagnates for larger sample sizes. {The proposed method outperforms other methods.} The NGP-based approaches outperform the one GP-based approaches under the same acquisition function. This is a very direct result of the more accurate predictions for the NGP model. 
}


\subsection{ 1-d non-GP model}
\label{1dngp}
Suppose the inner computer model is
   \begin{equation}
 h( x)= (1+|x|)^{-4}, x\in[-1,1],\nonumber
  \end{equation}
  and the outer computer model is
     \begin{equation}
 g( h)=h\sin(7\pi h/2).\nonumber
  \end{equation}
 The global minimum of  this nested computer experiment  is  $ (0,-1)$.  Fig.  \ref{Figure levy ngp} compares the performance of one GP model and NGP model with $n_0=10$. 
 
\begin{figure}[!ht]
\centering
\includegraphics[width=\hsize]{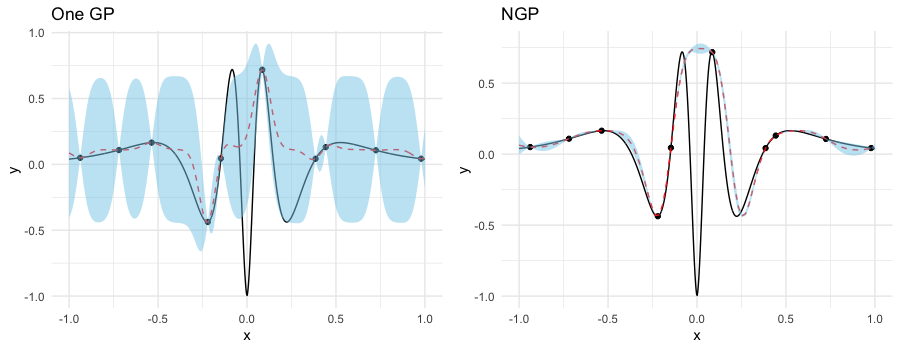}
\caption{{Left:  predictions (red dotted line) and 95\% confidence intervals  of the one-GP model build by using $(\tilde X_{n_0},Y_{n_0})$, with $n_0=10$;  Right: predictions (red dotted line) and 95\% confidence intervals of the NGP model.}}
 \label{Figure levy ngp}
\end{figure}

Fig.  \ref{Figure levy ngp} shows that both the one-GP model and the NGP model perform poor in $x\in [-0.1,0.1]$. The reason is that, values of the true function change fast in  $x\in [-0.1,0.1]$, but the design is sparse in $ [-0.1,0.1]$. Except  at the points that belong to $ [-0.1,0.1]$, the NGP model outperforms  the one-GP model.

Via the $3$-fold CV test, we can find that the NGP model is not Gaussian. Therefore, in the NBO algorithm, the sequential point  is collected by maximizing ${\rm NEI}_n$ (\ref{ngp-ei}). Set $L_t=-10\sigma_1(x)+\mu_1( x)$ and $U_t=10\sigma_1( x)+\mu_1(x)$, MCMC method is used to evaluate  (\ref{ngp-ei}) and the EQI function. { The log-optimality gaps against the number of samples for the six methods are shown in Fig.} \ref{Figure levy}. 
\begin{figure}[!ht]
\centering
\includegraphics[width=\columnwidth]{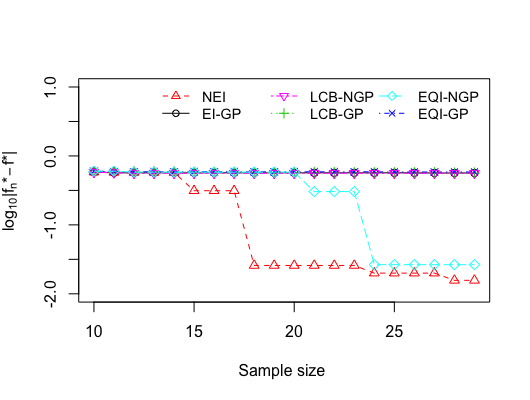}
\caption{Average optimality gap over $50$ replications by different methods.}
\label{Figure levy}
\end{figure}

{ From Fig.} \ref{Figure levy}{, {we can conclude that the optimality gaps for NEI and
EQI-NGP enjoy steady improvements as $n$
increases.}   However, the other methods fall into a local optimal point, which is included in the initial design. 
This
shows that the proposed method balances the optimal point of the fitted model with the exploration of other regions.

It is worth noting that, since the LCB depends only on the posterior mean and variance of $f(\tilde {\bm x})$, this  acquisition function lose its advantage when the posterior distribution of  $f(\tilde {\bm x})$ is non-Gaussian. }

\subsection{ 4-d GP model}
Suppose the inner computer model includes two functions: the three-hump camel function

\begin{equation}
h_1({\bm x})=2x_1^2-1.05x_1^4+x_1^6/6+x_1x_2+x_2^2, \nonumber
\end{equation}

and the six-hump camel function
 
\begin{align*}
h_2({\bm x})=(4-2.1x_3^2+x_3^4/3)x^2_3+x_3x_4+(-4+4x_4^2)x_4^2,
\end{align*}
Here, ${\bm x}=(x_1,x_2,x_3,x_4)\in [-1,1]^4$.
Suppose the outer computer model is the Branin function

\begin{equation}
g(\bm h)=\frac{1}{51.95}\left[ g_1(\bm h)+(10-\frac{10}{8\pi})\cos(\bar h_1)-44.81\right],\nonumber
\end{equation}
where $g_1(\bm h)=(\bar h_2-\frac{5.1\bar h_1^2}{4\pi^2}+\frac{5\bar h_1}{\pi}-6)^2$, $\bar h_1=5(h_1-1)$, $\bar h_2=5(h_2+1)$.
The global minimum of $f=g(\bm h({\bm x}))$ is at ${\bm x}^*=(-0.121, 0.547, 0.915, 0.715)$ and the corresponding function value  is $-16.644$.
Let $n_0=40$, we still use the maximin Latin hypercube design to collect data. Then we build GP models for inner and outer computer models. Via the $3$-fold  CV test, we have that the NGP model is Gaussian.

\begin{figure}[!ht]
\centering
\includegraphics[width=\columnwidth]{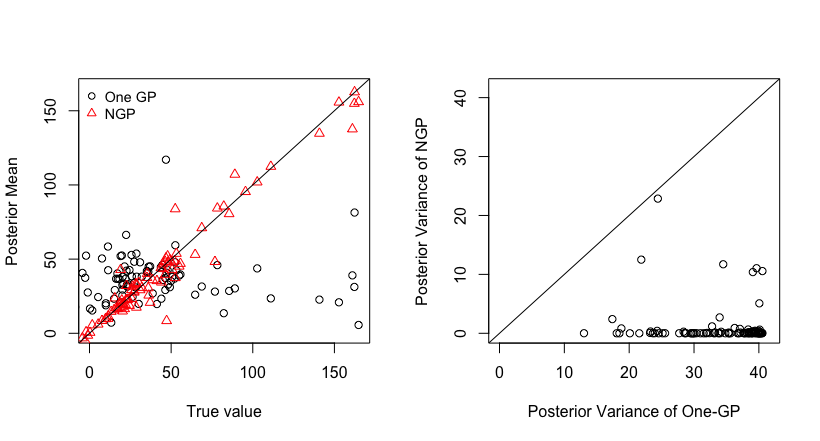}
\caption{Left: Posterior mean of the one-GP (black circles) and NGP (red triangles);  Right: Posterior variance of the one-GP  and NGP.}
\label{Figure branin ngp}
\end{figure}

Fig.  \ref{Figure branin ngp} compares the prediction performance of one GP model and NGP model at $100$ un-observed  locations. These $100$ testing locations are sampled by  the  maximin Latin hypercube design. Left of  Fig.  \ref{Figure branin ngp} shows the comparison between  predictions of  different models and the true  outputs of the nested computer experiment. We see that, predictions given by  the NGP model at these testing locations are much closer to the true values.   The $100$ points (black circles) in Fig.  \ref{Figure branin ngp} right compare the posterior variances given by the one-GP model and the NGP model. Because all $100$ points are under the line ``$y=x$", it indicates that posterior variances given by 
the NGP model are smaller than  posterior variances given by the one-GP model.

Fig. \ref{Figure bias} shows the log-optimality gap $log_{10}(f_n^*-f^*)$ against the number of samples $n$. Results of the log-optimality gap are averaged over 50 replications.

\begin{figure}[!ht]
\centering
\includegraphics[width=\columnwidth]{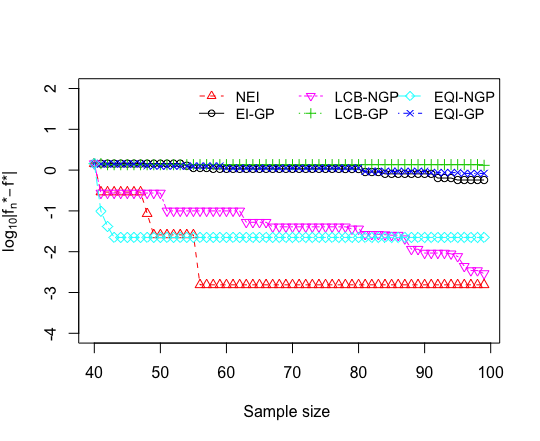}
\caption{Average optimality gap over 50 replications by different methods.}
\label{Figure bias}
\end{figure}

{We can see from Fig. } \ref{Figure bias}{ that the proposed method outperforms other methods: the optimality gap for the latter methods stagnates for larger sample sizes, whereas the former enjoys steady improvements as $n$ increases.} 

In summary, results of the numerical simulations show that {the proposed NBO method has three advantages}: (i) it incorporates the nested structure information and makes full use of the inner computer model outputs; (ii) it improves the prediction accuracy significantly; (iii) it avoids the convergence to local minimum and identifies the global optimum more efficiently.

\section{Case Study via Composite Structures Assembly}
 \label{case}
 {Composite structures have become increasingly used in many major products (e.g., fuselages, wings, car bodies, solar panels, spacecraft) due to their superior characteristics} including high strength-to-weight ratio, high stiffness-to-weight ratio, potential long life, and low life-cycle cost. However, fabrication deviations are inevitable in composite structures. It is timely important to address the quality control in composite structures assembly. 
 
 \begin{figure}[h!]
\centering
\includegraphics[width=\columnwidth]{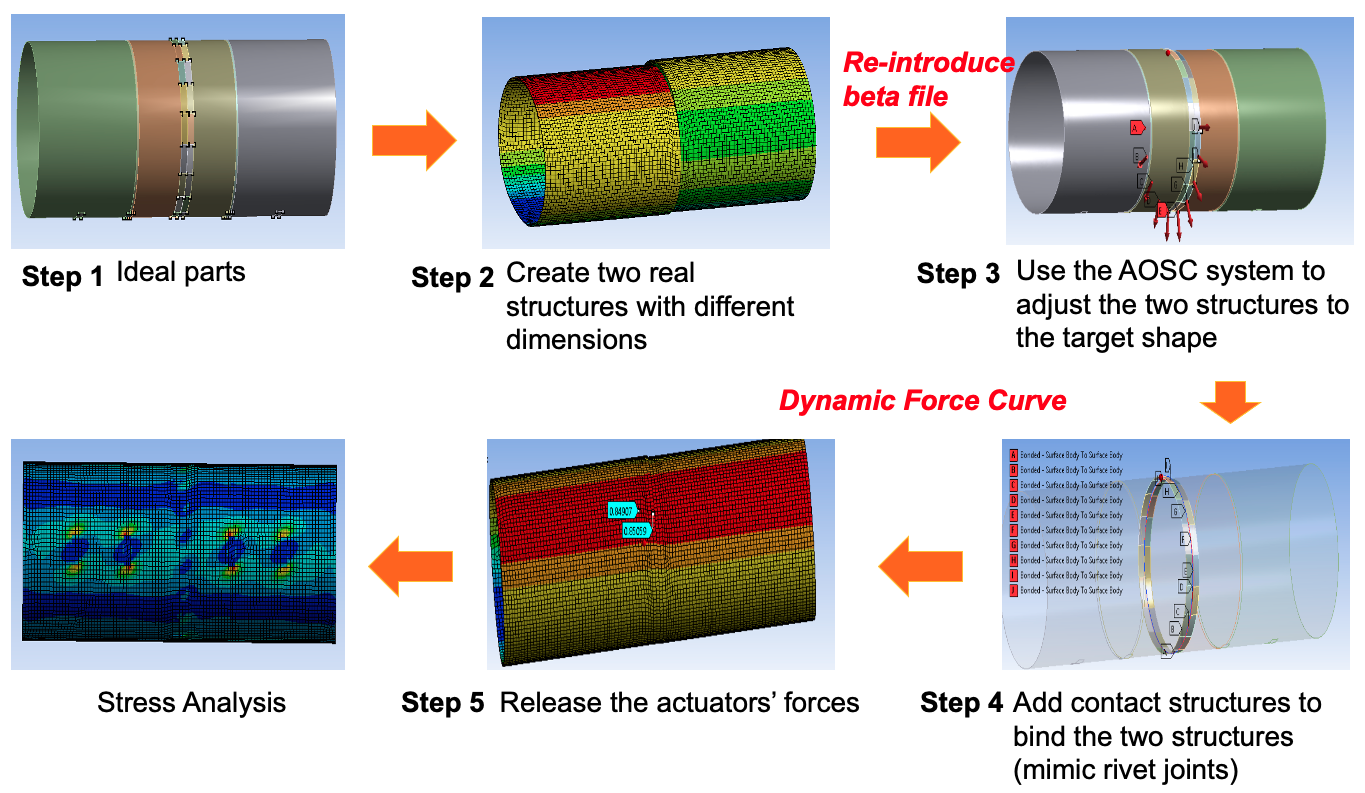}
\caption{The computer experiment mimics the composite structures assembly process.}
\label{Figure flow}
\end{figure}
 
 One digital twin simulation platform for composite structures assembly was developed to mimic the fabrication process of carbon-fiber reinforced composites \cite{wen2018feasibility,wen2019virtual}. This computer simulation platform was built based on ANSYS PrepPost Composites workbench, and it was calibrated and validated via physical experiments. The calibration process refers to \cite{wang2020effective}. The digital twin simulation can conduct virtual assembly to illustrate detailed composite structures joint. As shown in Fig. \ref{Figure flow}, the virtual assembly simulation includes multiple steps: (i) generate composite structures with deviations, (ii) apply Automatic Optimal Shape Control technique \cite{yue2018surrogate} to adjust the dimensions; (iii) add revit joins and then release actuators' forces; (iv) do dimensional analysis and stress analysis. 
 
 This multistep computer simulation for composite structure assembly has nested structure. As shown in Fig. \ref{Figure nest}, the inner computer model simulates the shape control of a single composite structure. The automatic optimal shape control can adjust the dimensional deviations of one composite fuselage and make it align well with the other fuselage to be assembled. The outer computer model simulates the process of composite structures assembly, where the inputs are critical dimensions from two parts, and the outputs are internal stress after assembly. Table     \ref{tab:com1} summarizes the inputs and outputs information in computer experiments. We will conduct nested Bayesian optimization for this nested computer experiment to identify the optimal assembly that can minimize the residual stress after assembly.

\begin{figure}[!ht]
\centering
\includegraphics[width=\columnwidth]{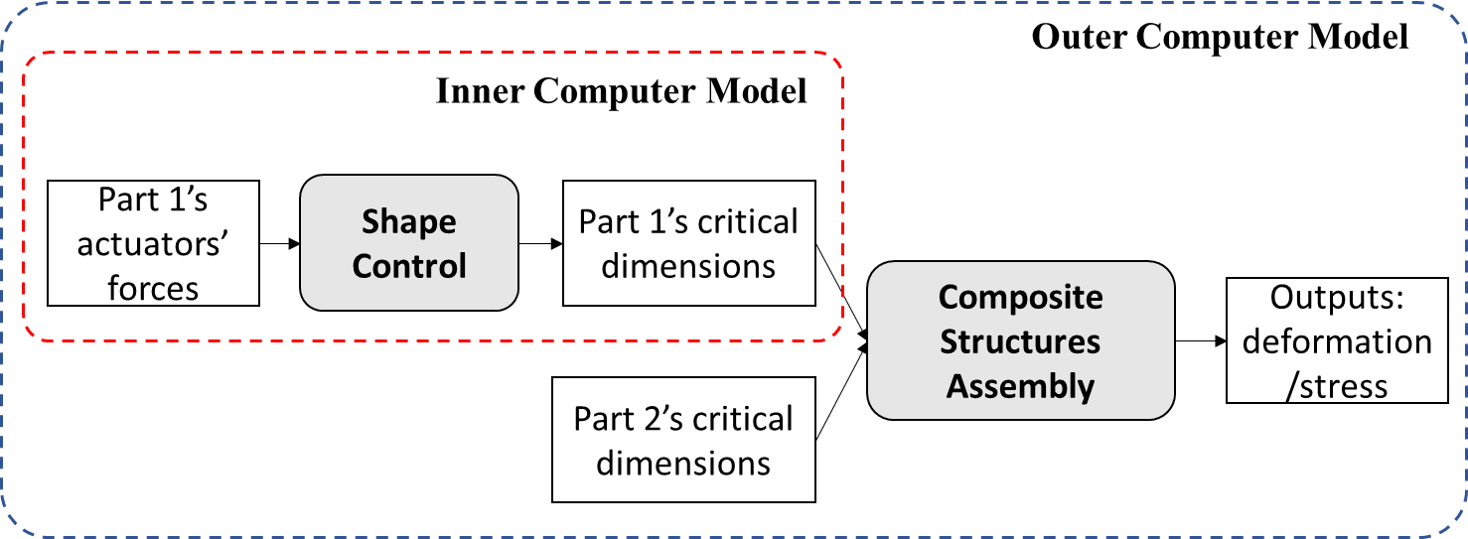}
\caption{Nested computer experiments in composite structures assembly.}
\label{Figure nest}
\end{figure}
\begin{table*}[!ht]
    \centering
    \caption{Inputs and outputs for the nested computer experiments}
    \label{tab:com1}
    \setlength{\tabcolsep}{5pt}
    \renewcommand{\arraystretch}{2}
 {  \begin{tabular}{lccccccc}
        \hline
        Inner computer model\\
          &Name of variable& Dimension &Range of values\\
          \hline
      Inputs&  Part 1’s actuators’ forces (${ x}$)  & $10$  &${ x}_{i}\in (-250,250), i=1,\ldots,10$ \\
        Outputs&  Part 1’s critical dimensions  ($\bm h({ x})$) & $5$  \\
               \hline
              Outer computer model\\
                     &Name of variable& Dimension &\\
          \hline
            Inputs&         Part 1’s critical dimensions  ($\bm h({ x})$) & $5$& \\
            &     Part 2’s critical dimensions (${ x'}$)  & $5$& \\
            Outputs&  Mean of  Stress & $1$  \\
                 \hline
    \end{tabular}}
   \end{table*}

Let $n_0=100$,  we collect the inner computer model outputs  $H_{n_0}$ on a  maximin Latin hypercube design $X_{n_0}$, and the outer computer model outputs $Y_{n_0}$ on $(H_{n_0},X_{n_0})$. 
 We conduct the $2$-fold CV test and find that the NGP model is non-Gaussian. We split  the initial data into 70\% as training and 30\% as a testing set randomly, and use the training data to build the GP and NGP models. The testing data is used to compare the prediction accuracy of different models. 

\begin{figure}[!ht]
\centering
\includegraphics[width=\columnwidth]{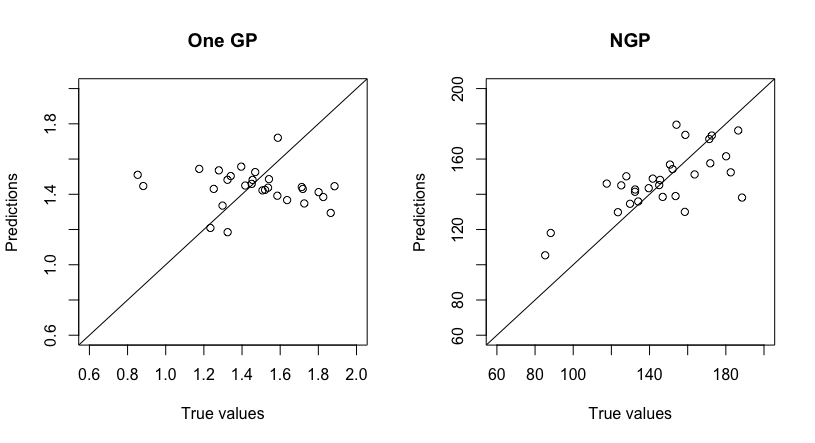}
\caption{Predictions given by the GP (left) and NGP (right) v.s. the true outputs;.}
\label{Figure pred}
\end{figure}

Fig. \ref{Figure pred} shows that the NGP model outperforms the one-GP model. 
Because the dimension of the inputs is $15$, it is time-consuming to search the optimal point of EI and NEI function in Bayesian optimization. Following \cite{chen2020rbfbo}, instead of directly optimize the acquisition functions  over $\mathcal{X}$, we choose a set of candidate point $\mathcal{X}_{cand}$ from the whole search domain and then find the next point in $\mathcal{X}_{cand}$.
In this work, we  select $\mathcal{X}_{cand}$ on a maximin Latin hypercube design and the sample size of $\mathcal{X}_{cand}$ is set to be $1000$. { Let $N=200$,  Fig.}\ref{Figure case opt} {shows the optimal results given by  different methods.}
\begin{figure}[!ht]
\centering
\includegraphics[width=\columnwidth]{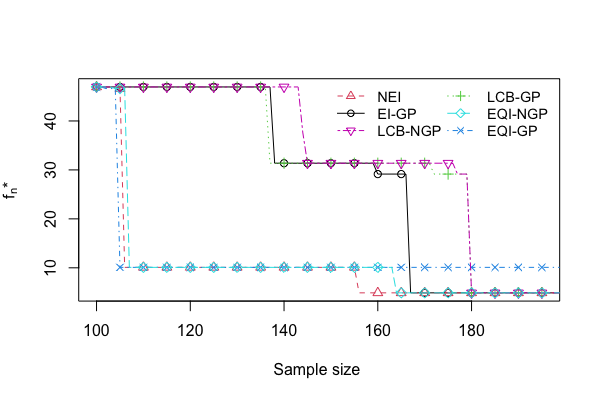}
\caption{The optimal results given by  different methods.}
\label{Figure case opt}
\end{figure}
{From Fig.}\ref{Figure case opt}{, we have that except for the EQI method under one-GP model, the others obtain the same minimum of residual stress with  4.885 psi (pound per square inch). Moreover, the proposed method identifies this residual stress with a minimum number of sequential points, which indicates the high effectiveness of the proposed method.}



\section{{Summary and Discussions}}
\label{sec.dis}

Computer experiments and digital twins have ubiquitous influence on engineering systems. Since the multi-step simulations or hierarchical structure of systems, many computer experiments have nested structures. This paper proposed a novel Bayesian optimization {method} for nested computer experiments. We first derived the nested Gaussian process models to serve as surrogates for the computer models. We proved the distribution of nested outputs given it is Gaussian or non-Gaussian. We also deduced the closed forms of nested expected improvement, and proposed one new algorithm for nested Bayesian optimization. The proposed NBO method can make full use of the nested structure and intermediate outputs to identify the global optimum efficiently. It avoids convergence to the local optimum which may occur in standard Bayesian optimization. We validated the performance of NBO based on three numerical studies and one case study. In the case study, the proposed NBO can minimize the residual stress for composite structures assembly, and achieve a much better result than the conventional Bayesian optimization {methods}. 

{The proposed method may be faced with generalizability challenge when the system has multiple connected models. Specifically, approximating the multiple nested computer models by a suitable surrogate model needs to estimate more hyperparameters. More training samples will be required for accurate parameter learning. High-dimensionality of parameters may result in high computational cost of Bayesian optimization. Furthermore, the fitting multiple connected computer models by a nested GP may have non-identifiability issue. In future research, we will investigate the identifiability conditions and new nested Bayesian optimization {methods} for complex multiple connected systems.} 



\section*{Acknowledgment}

 Dr. Wang's research was supported by the National Natural Science 
Foundation of China (12101024), the Natural Science Foundation
of Beijing Municipality (1214019).

\bibliographystyle{IEEEtran}
\bibliography{Ref.bib}

\clearpage 

\appendices
\section{Gaussian Process models}
\label{gp model}
In this section, we introduce GP models to mimic the inner and outer computer outputs. Suppose
\begin{equation}
\begin{aligned}
h_k(\cdot)&=\mu_{h_k}(\cdot)+Z_{h_k}(\cdot), k=1,\ldots,p,\\
  \mu_{h_k}(\cdot)&=\bm b_{h_k}^T(\cdot)\bm\beta_{h_k},   Z_{h_k}(\cdot)\sim GP(\bm {0}, \sigma_{h_k}^2\Phi_{h_k});\\
g(\cdot)&=\mu_g(\cdot)+Z_g(\cdot), \\
 \mu_g(\cdot)&=\bm b_g^T(\cdot)\bm\beta_g,  Z_g(\cdot)\sim GP(0,\sigma_g^2\Phi_g).
\end{aligned}
\end{equation}
For the $k$th output of inner computer model, $\bm b_{h_k}(\cdot)$ consists of $q_{h_k}$ basis functions for the mean
function $\mu_{h_k}$; $\bm\beta_{h_k}$ denotes its corresponding coefficients, and $GP(0,\sigma_{h_k}^2\Phi_{h_k})$ denotes a stationary Gaussian Process with mean zero, variance $\sigma_{h_k}^2 $ and correlation function $\Phi_{h_k}(\cdot)$. 
For the outer computer model, $g(\cdot)\sim  GP(\bm b_g^T(\cdot)\bm\beta_g,\sigma_g^2\Phi_g)$, where $\bm b_g$ and $\bm\beta_g$ are $q_g\times 1$ vectors; $\sigma_g^2$ is the process variance and  $\Phi_g(\cdot)$ is the correlation function.
Common choices of $\Phi_{h_k}(\hbar)$ and $\Phi_g(\hbar)$ include the Gaussian correlation functions 
 \begin{eqnarray}
\exp(-\theta \hbar^2),
\label{exp}
\end{eqnarray}
and the Mat\'ern correlation functions with
 \begin{eqnarray}
\frac{1}{\Gamma(\nu)}\left(\frac{2\sqrt{\nu}\hbar}{\theta}\right)^{\nu}K_{\nu}\left(\frac{2\sqrt{\nu}\hbar}{\theta}\right),
\label{matern}
\end{eqnarray}
where $\hbar\geq 0$ is a distance between two inputs of the GP model. $\theta>0$ is the correlation parameter and  $K_{\nu}$ denotes the modified Bessel function of the second kind with order ${\nu}$.

Denote ${\rm{ {\Phi}}}_{h_k}=(\Phi_{h_k}(\parallel x_{i}-x_{j}\parallel))_{i,j=1}^n$;  ${\bm\phi}_{h_k}(x)=(\Phi_{h_k}(\parallel x-x_{i}\parallel))_{i=1}^n$ and $\mathbf{B}_{h_k}=\left(b_{h_k}( x_{1}),\ldots,b_{h_k}(x_{n})\right)^T$. The posterior distribution of
$h_k(\cdot)$ at an unobserved input $x$ has the closed form \cite{santner2018design}:
\begin{equation}
\label{posnesth}
h_k(x)|H_n, \tilde X_n\sim N(\hat { h}_{k,n}(x),s_{h_k}^2( x)).
\end{equation}
Here, the posterior mean is
\begin{equation}
\label{blue}
\hat {h}_{k,n}( x)=\bm b_{h_k}^T(x)\hat{\bm\beta}_{h_k}+{\bm\phi}_{h_k}^T( x){\rm{\Phi}}_h^{-1}(H_{k,n}-\mathbf{B}_{h_k}\hat{\bm\beta}_{h_k}),
\end{equation}
where $\hat{\bm\beta}_{h_k}=(\mathbf{B}^T_{h_k}{\rm{\Phi}}_{h_k}^{-1}\mathbf{B}_{h_k})^{-1}\mathbf{B}^T_{h_k}{\rm{\Phi}}_{h_k}^{-1}H_{k,n}$, 
and the posterior variance is

\begin{eqnarray}
\begin{aligned}
\label{var}
s^2_{h_k}(x)&=\sigma_{h_k}^2 \left\{\Phi_{h_k}( x, x)-{\bm\phi}_{h_k}^T( x){\rm{\Phi}}_{h_k}^{-1}{\bm\phi}_{h_k}( x)\right\}\\
&+\sigma_{h_k}^2U^T_{h_k}(\Tilde{\bm x})(\mathbf{B}^T_{h_k}{\rm{\Phi}}_{h_k}^{-1}\mathbf{B}_{h_k})^{-1}U_{h_k}(\Tilde{\bm x}),
\end{aligned}
\end{eqnarray}
where $U_{h_k}(x)=\bm b_{h_k}(x)-\mathbf{B}^T_{h_k}{\rm{\Phi}}_{h_k}^{-1}{\bm\phi}_{h_k}(x)$. Formulations of the posterior mean and posterior variance function of $g(\cdot)|D_n$ are the same as (\ref{blue})  and (\ref{var}), respectively.  In addition, the process variance $\sigma_{h_k}^2$ and the hyper-parameter $\theta$ in the correlation function are always unknown in practice, maximum likelihood estimators (MLEs)  can be plugged into (\ref{posnesth}) to obtain the posterior distribution of $h_k$.


\section{Technical Proofs}\label{App:proof}

\begin{proof}[Proof of Theorem \ref{lemma2}]

Stochastic Taylor expansion of (\ref{nested gp}) shows that
 \begin{equation}
 \begin{aligned}
 f(\Tilde{\bm x})|D_n=Z(\Tilde{\bm x})+Re(\Tilde{\bm x}),
 \end{aligned}
\end{equation}
where
\begin{equation}
\small
\begin{aligned}
\label{nested gp1}
Z(\Tilde{\bm x})
&= \mu_{Z}(\Tilde{\bm x}) +\bm c^T_h(\Tilde{\bm x})\bm \xi_h+ c_g(\Tilde{\bm x})\xi_g+\bm c^T_{h,g}(\Tilde{\bm x})\bm \xi_h \xi_g,\\
\end{aligned}
\end{equation}
with $\mu_{Z}(\Tilde{\bm x})=\hat g_n(\hat {\bm h}^T_n(x),x')$  the global trend of $Z(\Tilde{\bm x})$; $\bm c^T_h(\Tilde{\bm x})=\frac{\partial \hat g_n}{\partial \bm h}(\hat {\bm h}^T_n(x),x')\bm s_h(x)$ the uncertainty in $Z(\Tilde{\bm x})$  due to the inner GP model; $c_g(\Tilde{\bm x})=s_g(\hat {\bm h}^T_n(x),x')$ the uncertainty in $Z(\Tilde{\bm x})$  due to the outer GP model; $\bm c^T_{h,g}(\Tilde{\bm x})=\frac{\partial  s_g}{\partial \bm h}(\hat {\bm h}^T_n(x),x')\bm s_h(x)$ the uncertainty arising from the combined effect of the inner and outer  models.

 $Re(\Tilde{\bm x})$ is the Lagrange remainder, which presents the approximation error between $f(\Tilde{\bm x})|D_n$ and $Z(\Tilde{\bm x})$.
From Corollary 2 in \cite{yang2021note}, we have that, by assuming   the second order derivatives of $\hat g_n$ and $s_g$ with respect to $\bm h$  are uniformly bounded, the Lagrange's error bound is $Re(\Tilde{\bm x})=O\left(\left[\sum_{k=1}^p [\bm s_h(x)\bm\xi_h]_k\right]^2\right)$,
 {where $[\bm s_h(x)\bm\xi_h]_k$ is the $k$th element of $\bm s_h(x)\xi_h$.
  Proposition 3.2 in} \cite{wang2021inference} {shows that $\lim_{n\rightarrow \infty}\sup_x \bm s_{h}(x)\bm\xi_h$ converges to $\bm 0$ in probability, and thus the distribution of $Z(\Tilde{\bm x})$ converges to the distribution of $f(\Tilde{\bm x})|D_n$. }
As a result,  the nested computer models $f(\Tilde{\bm x})|D_n$ can be approximated by $Z(\Tilde{\bm x})$.

Next, we focus on the distribution of $Z(\Tilde{\bm x})$.
 Denote $\eta_h$ to be a standard normal random variable, where the subscript  $h$ indicates that this randomness is caused by the inner GP model. {Let $s^2_1(\Tilde{\bm x})={\bm c^T_{h}(\Tilde{\bm x})\bm c_{h}(\Tilde{\bm x})}$, $s^2_2(\Tilde{\bm x})={\bm c^T_{h,g}(\Tilde{\bm x})\bm c_{h,g}(\Tilde{\bm x})}$. }
{ By some numerical calculations, it is easily verified that
$Z(\Tilde{\bm x})$ can be represented as }
\begin{align*}
Z(\Tilde{\bm x}) &= \left\{\eta_h+\frac{c_g(\Tilde{\bm x})}{s_2(\Tilde{\bm x})}\right\} \times \left\{s_2(\Tilde{\bm x})\xi_g+s_1(\Tilde{\bm x})\right\} \\ 
&+  \left[\mu_{Z}(\Tilde{\bm x})-c_g(\Tilde{\bm x})\frac{s_1(\Tilde{\bm x})}{s_2(\Tilde{\bm x})}\right].\\
&=Z_1(\Tilde{\bm x})Z_2(\Tilde{\bm x})+z_0(\Tilde{\bm x}).
\end{align*}

 {Then  for fixed $\Tilde{\bm x}\in\mathcal {X}$, $Z_1(\Tilde{\bm x})-z_0(\Tilde{\bm x})$ is a random variable generated by a production of two normal variables  $Z_1(\Tilde{\bm x})$ and $Z_2(\Tilde{\bm x})$, with $Z_1(\Tilde{\bm x})\sim N(\mu_1(\Tilde{\bm x}),\sigma^2_1(\Tilde{\bm x}))$ and $Z_2(\Tilde{\bm x})\sim N(\mu_2(\Tilde{\bm x}),\sigma^2_2(\Tilde{\bm x}))$. Here, $\mu_1(\Tilde{\bm x})={c_g(\Tilde{\bm x})}/\sqrt{\bm c^T_{h,g}(\Tilde{\bm x})\bm c_{h,g}(\Tilde{\bm x})}$; $\mu_2(\Tilde{\bm x})=\sqrt{\bm c^T_{h}(\Tilde{\bm x})\bm c_{h}(\Tilde{\bm x})}$; $\sigma^2_1(\Tilde{\bm x})=1$; $\sigma^2_2(\Tilde{\bm x})={\bm c^T_{h,g}(\Tilde{\bm x})\bm c_{h,g}(\Tilde{\bm x})}$.}

The exact probability density function of $Z(\Tilde{\bm x})$ can be 
computed as \cite{2016Exact}:
\begin{equation}
\label{exactpdf}
\begin{aligned}
p_Z(z)=\frac{\int_{-\infty}^{\infty}\frac{1}{|t|}\exp\left\{-\frac{u_1^2(t,\Tilde{\bm x})+u^2_2(z,t,\Tilde{\bm x})}{2}\right\}dt}{2\pi\sigma_1( \Tilde{\bm x})\sigma_2(\Tilde{\bm x})},
 \end{aligned}
\end{equation}
where $u_1(t,\Tilde{\bm x})=\frac{t-\mu_1(\Tilde{\bm x})}{\sigma_1(\Tilde{\bm x})}$, $u_2(z,t,\Tilde{\bm x})={\frac{z+z_0(\Tilde{\bm x})-t\mu_2(\Tilde{\bm x})}{|t|\sigma_2(\Tilde{\bm x})}}$.
The cumulative density function of $Z(\Tilde{\bm x})$ is 

\begin{equation}
\label{cdf.2}
\begin{aligned}
P_Z(z)=\int_{-\infty}^{\infty}\frac{1}{\sigma_1(\Tilde{\bm x})}\phi_N\left(u_1(t,\Tilde{\bm x})\right)\Phi_N{\left(u_2(z,t,\Tilde{\bm x})\right)}d t,
\end{aligned}
\end{equation}
where ${\rm\Phi}_N$ is the cumulative distribution function of the standard normal distribution and $\phi_N$ is the probability density function.

Because $Z_1(\Tilde{\bm x})$ is independent from $Z_2(\Tilde{\bm x})$, mean function and variance function of $Z(\Tilde{\bm x})$ can be easily deduced:
\begin{equation}\nonumber
\small
\begin{aligned}
{\rm E}[Z(\Tilde{\bm x})]&={\rm E}[Z_1(\Tilde{\bm x})]\times {\rm E}[Z_2(\Tilde{\bm x})]+z_0(\Tilde{\bm x}),\\\nonumber
&=\mu_1(\Tilde{\bm x})\mu_2(\Tilde{\bm x})+z_0(\Tilde{\bm x}).\\\nonumber
{\rm Var}[Z(\Tilde{\bm x})]
&={\rm E}[Z^2_1(\Tilde{\bm x})]\times {\rm E}[ Z^2_2(\Tilde{\bm x})]-\mu^2_1(\Tilde{\bm x})\mu^2_2(\Tilde{\bm x}),\\
&=\left[\sigma_1^2(\Tilde{\bm x})+\mu^2_1(\Tilde{\bm x})\right] \left[\sigma_2^2(\Tilde{\bm x})+\mu^2_2(\Tilde{\bm x})\right]-\{\mu^2_1(\Tilde{\bm x})\mu^2_2(\Tilde{\bm x})\},\nonumber
\end{aligned}
\end{equation}
which implies the desired results.
\end{proof}

\begin{proof}[{Derivation of the NEI acquisition function}  (\ref{ngp-ei})]
Based on the definition of NEI function (\ref{nei}), we have
\begin{equation}
\label{nei.0}
\begin{aligned}
{\rm NEI}_n(\Tilde{\bm x})\approx &{\rm E}_{Z}(f_n^*-Z(\Tilde{\bm x}))_{+}.
\end{aligned}
\end{equation}
Let $U=f_n^*-Z(\Tilde{\bm x})$, we can rewrite ${\rm NEI}_n(\Tilde{\bm x})$ as
\begin{equation}
\label{nei.1}
\begin{aligned}
\frac{1}{\sqrt{2\pi}\sigma_1(\Tilde{\bm x})}\int_{-\infty}^{\infty}\exp\left\{-\frac{1}{2}\times \frac{(t-\mu_1(\Tilde{\bm x}))^2}{\sigma_1^2(\Tilde{\bm x})}\right\}B(t)dt,\\
\end{aligned}
\end{equation}
where 
$B(t)=\int_0^{\infty} \frac{u}{|t|\sqrt{2\pi}\sigma_2(\Tilde{\bm x})} \exp\left\{-\frac{1}{2}\times \left[u'_2-\frac{\mu_2(\Tilde{\bm x})}{\sigma_2(\Tilde{\bm x})}\right]^2 \right\}du$,
and $u'_2=\frac{f_n^*-u-z_0(\Tilde{\bm x})-t\mu_2(\Tilde{\bm x})} {|t|\sigma_2(\Tilde{\bm x})}$. 
By some easy numerical calculations, we have that, 
\begin{equation}
\label{B}
\begin{aligned}
B(t)
 =&\left[f_n^*-z_0(\Tilde{\bm x})-t\mu_2(\Tilde{\bm x})\right]\Phi_N\left(u_2(f_n^*,t,\Tilde{\bm x})\right)+\\
 & |t|\sigma_2(\Tilde{\bm x}) \phi_N\left( {{u_2(f_n^*,t,\Tilde{\bm x})}}\right).
 \end{aligned}
\end{equation}
Substituting (\ref{B}) into (\ref{nei.1}), 
the desired results then can be obtained.

\end{proof}

\section{{Choice of the correlation functions}}
\label{rbfs}

{In this section, we illustrate the reasons that we choose Gaussian correlation functions} (\ref{exp}) { as the correlation functions in the numerical studies.}

{We compared the prediction accuracy of the one-GP model and the NGP model in the example} \ref{1dgp}{ with different correlation functions}:
\begin{itemize}
    \item {Gaussian (Radial Basis Function) Kernel (\ref{exp})\;}
    \item {Exponential Kernel: $\exp(-\theta\hbar)$;}
    \item {Power-exponential Kernel: $\exp(-\theta\hbar^p), p>0$;}
    \item {Mat\'ern correlation function} (\ref{matern}){ with $\nu=\frac{3}{2}$;}
    \item {Mat\'ern correlation function  with $\nu=\frac{5}{2}$.}
\end{itemize}

{Here, $\hbar \geq 0$ is a distance between two inputs of the GP model. $\theta>0$ is the correlation parameter which can be estimated by the maximum likelihood method. Figure} \ref{Figure kernel1}-  \ref{Figure kernel5} {compare the accuracy  of  one-GP and NGP models with different correlation functions.}

\begin{figure}[!ht]
\centering
\includegraphics[width=\columnwidth]{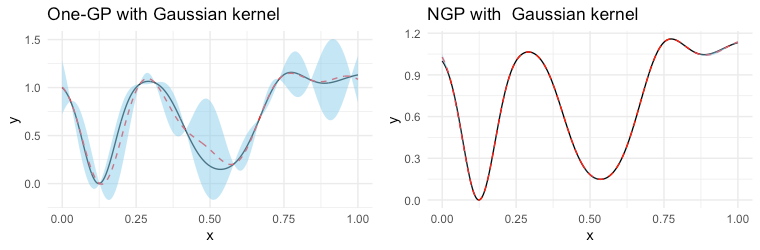}
\caption{{Left:  predictions (red dotted line) and 95\% confidence intervals  of the one-GP model build by using $(\tilde X_{n_0},Y_{n_0})$, with $n_0=10$;  Right: predictions (red dotted line) and 95\% confidence intervals of the NGP model.}}
\label{Figure kernel1}
\end{figure}

\begin{figure}[!ht]
\centering
\includegraphics[width=\columnwidth]{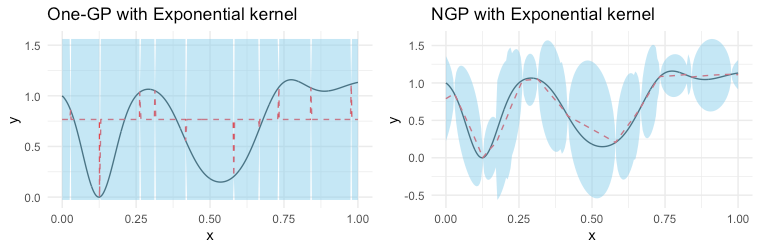}
\caption{{Predictions (red dotted line) and 95\% confidence intervals (blue interval) of the one-GP model (left) and the NGP model (right).}}
\label{Figure kernel2}
\end{figure}

\begin{figure}[!ht]
\centering
\includegraphics[width=\columnwidth]{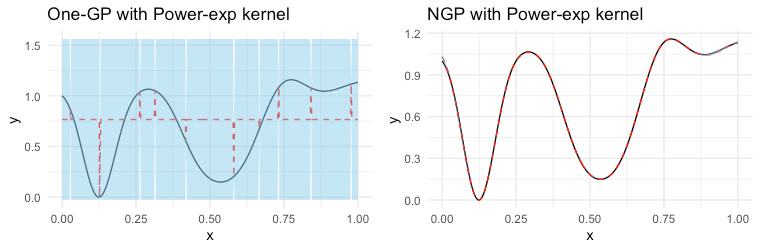}
\caption{{Predictions (red dotted line) and 95\% confidence intervals (blue interval) of the one-GP model (left) and the NGP model (right), with the parameter $p=1.96$.}}
\label{Figure kernel3}
\end{figure}

\begin{figure}[!ht]
\centering
\includegraphics[width=\columnwidth]{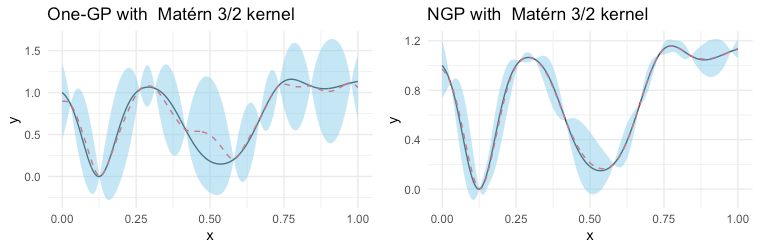}
\caption{{Predictions (red dotted line) and 95\% confidence intervals (blue interval) of the one-GP model (left) and the NGP model (right).}}
\label{Figure kernel4}
\end{figure}

\begin{figure}[!ht]
\centering
\includegraphics[width=\columnwidth]{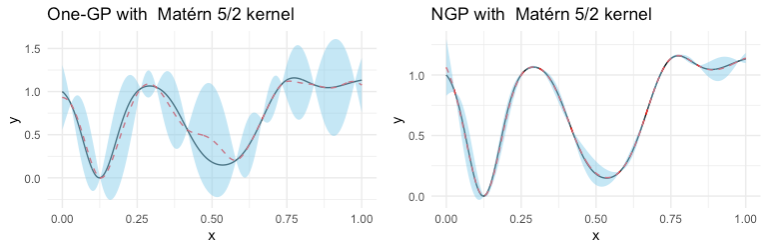}
\caption{{Predictions (red dotted line) and 95\% confidence intervals (blue interval) of the one-GP model (left) and the NGP model (right).}}
\label{Figure kernel5}
\end{figure}

{From Figure} \ref{Figure kernel1}-  \ref{Figure kernel5}{, we can see that, with the same kernel, the NGP model outperforms the one-GP model.  By taking the prediction accuracy of both the one-GP and the NGP models into accout, we recommend choosing the Gaussian correlation functions} (\ref{exp}). 


\end{document}